\documentclass[conference,a4paper]{IEEEtran}

\usepackage{amsthm}
\usepackage{amssymb}
\usepackage{graphicx}
\usepackage{epstopdf}
\usepackage[utf8]{inputenc} 
\usepackage[T1]{fontenc}
\usepackage{url,balance}
\usepackage{ifthen}
\usepackage{cite}
\usepackage[cmex10]{amsmath} 

\theoremstyle{definition}
\newtheorem{definition}{\textbf{Definition}}
\theoremstyle{example}

\newtheorem{example}{Example}
\theoremstyle{theorem}
\newtheorem{theorem}{Theorem}[section]

\theoremstyle{corollary}

\newcommand\scalemath[2]{\scalebox{#1}{\mbox{\ensuremath{\displaystyle #2}}}}

\theoremstyle{lemma}
\newtheorem{lemma}{Lemma}[section]

\theoremstyle{remark}

\theoremstyle{lemma}

\begin{document}

\title{
	Multicast Networks Solvable over Every Finite Field}

\author{\IEEEauthorblockN{Niranjana Ambadi}
	\IEEEauthorblockA{Department of Electrical Communication Engineering\\
		Indian Institute of Science,
		Bangalore, India. 560012.\\
		Email: ambadi@iisc.ac.in}

}

\maketitle

\begin{abstract}
 In this work, it is revealed that an acyclic multicast network that is scalar linearly solvable over Galois Field of two elements, $GF(2)$,  is solvable over all higher finite  fields.
    An algorithm which, given a $GF(2)$ solution for an acyclic multicast network, computes the solution over any arbitrary finite field is presented. The concept of multicast matroid is introduced in this paper. Gammoids and their base-orderability  along with the regularity of a binary multicast matroid are used to prove the results.  
\end{abstract}

\section{Introduction}

\par A multicast network (henceforth denoted by $\mathcal{N}$), is a finite directed acyclic multigraph $G=(\mathcal{E;V})$ with a unique source node $s$ and a set of receivers $\mathcal{T}$. The source generates a set of $\omega$ symbols from a fixed symbol alphabet and will transmit them to all receivers through the network. The sets of input links and output  links  of  a  node $t$ are  denoted   by
$In(t)$ and $Out(t)$ respectively.  A  pair  of  links $(d,e)$ is  called  an adjacent pair if $d\in In(t)$ and $e\in Out(t)$ for  some $t\in \mathcal{V}$. The capacity of each link is assumed to be unity, i.e., one
symbol  can be transmitted on each link. For  a node $t$, the value of a maximum flow from the source node $s$ to node $t$ is denoted by $maxflow(t)$. Any receiver node $t\in \mathcal{T}$,  has at least $\omega$ edge-disjoint paths starting at the source $s$, i.e. $maxflow(t)\geq \omega ;\; \forall t \; \in \mathcal{T}$. $\mathcal{N}$ is solvable if all receivers can recover all the $\omega$ source symbols from their respective received symbols. When $|\mathcal{T}|>1$, routing is  insufficient to guarantee network solvability.

The idea of \textit{network coding} was first introduced by Ashwede et al. in \cite{ash}. The work stated that some network coding solution always exists over a sufficiently large alphabet.  It was further shown in \cite{LYC} that when the symbol alphabet is algebraically modelled
as a sufficiently large finite field, \textit{linear} network coding is  sufficient to yield
a solution. In a linear network code \cite{ITRY}, all the information symbols are regarded as elements of a finite field $\mathbb{F}$ called the \textit{base field}. These include the symbols that are generated by the  source as well as the symbols transmitted on the links. The symbol transmitted on an output link (an output
symbol) of a node is  a linear combination of all the symbols transmitted on the input links (input symbols)
of that node.
Encoding and decoding are based on linear algebra defined on the base field.

Koetter and Medard \cite{KoM} laid out an algebraic approach to network coding and proved that  a linear network code is guaranteed whenever the field size $q$ is larger than $\omega$ times the number of receivers, i.e., $q\geq \omega |\mathcal{T}|$. The field size requirement for the existence of a linear solution was further relaxed to $q>|\mathcal{T}|$ in \cite{HKMurota}. 
For single-source network coding, the Jaggi-Sanders algorithm \cite{jaggi} provides a construction of linear network codes that achieves the max-flow min-cut bound for network information flow.

Matroid theory is a branch of mathematics founded by Whitney \cite{Whitney}. It deals with the abstraction of various independence relations such as linear independence in vector spaces or the acyclic property in graph theory. In a network, the messages coming into a non-source node $t$, and the messages in the outgoing links of  a node $t$ are dependent. In  \cite{DouFreZeg}, this network form of dependence was connected with the matroid definition of dependence  and a general method of constructing networks from matroids was developed. Using this technique several well-known matroid examples were converted into networks that carry over similar properties.  

Kim and Medard \cite{Kim}  showed that a network is scalar-linearly solvable if and only if the network is a matroidal network associated with a representable matroid over a finite field. They proved that determining scalar-linear solvability of a network is equivalent to finding a representable matroid over a finite field and a valid network-matroid mapping.

Sun et al. \cite{sun} revealed that for a multicast network, a linear solution over a given finite field does not necessarily imply the existence of a linear solution over all larger finite fields. Their work showed that not only the field size, but the order of  subgroups  in  the  multiplicative  group  of  a  finite  field  affects the linear solvability. 

They proved that on an acyclic multicast network, if there is a linear solution
over $GF(q)$, it is not necessary that there is a linear solution over every $GF(q^*)$ with  $q^* \geq q$. 
In a multicast network, the coding vectors of a Linear Network Code (LNC) naturally induces a representable
matroid on the edge set, and in the strongest sense, this induced representable matroid is referred to as a network
matroid \cite{OLMLNC}, in which case the linear independence of coding vectors coincides with the independence structure of edge-disjoint paths \cite{SuLiChan}. Analogous to a network matroid, a gammoid in matroid theory characterizes the independence of node-disjoint paths in a directed graph. Concluding the discussions of the paper, Sun et al. \cite{sun}  conjectured that if a multicast network is linearly solvable over $GF(2)$, it is linearly solvable over all finite fields.
%

In this paper, it is proved that an acyclic multicast network that is linearly solvable over $GF(2)$ is linearly solvable over all finite fields. First,  a matroid called the \textit{multicast matroid} is introduced.  Further, it is proved that the multicast matroid is a gammoid.  Further, it is shown how any $\mathbb{F}$-linear solution for  $\mathcal{N}$ gives an $\mathbb{F}$-representation for the multicast matroid. It is then shown that  given a $GF(2)$ solvable multicast network, its multicast matroid  is  binary and base-orderable. The existence of solution over all fields follows from the regularity of this matroid.

The remainder of the paper is organized as follows:
\begin{itemize}
	\item{ Section \ref{conventions} introduces some notations and conventions that are used in this paper.}
	\item  In Section \ref{background} the following concepts are detailed: linear network codes,  linear multicast, matroid and its dual, graphic matroids, totally-unimodular matrix,  transversal matroids, gammoids, induced matroid of an LNC and series-parallel extension of a matroid.
	
	\item The novel concept of multicast matroids is defined and discussed in Section \ref{multmatsection}.
	\item Section \ref{baseordmult} discusses the base-orderability of  multicast matroids.

	\item In Section \ref{fundaresults2}, it is  proved that if a multicast network is solvable over $GF(2)$ it is solvable over all higher fields. 
	\item Given a multicast network that is solvable over $GF(2)$, it is possible to obtain the solution over an arbitrary field $\mathbb{F}$, starting with a solution over $GF(2)$. This algorithm is detailed in Section \ref{algorithm}.
	\item Section \ref{conclusion} concludes the discussions. 
\end{itemize}

\section{Conventions}
\label{conventions}
In a multicast network $\mathcal{N}$, a message generated at the source node $s$ consists of $\omega$ symbols in the base field $\mathbb{F}$. Let these be $x_1, x_2, \ldots, x_{\omega}$. This message is represented by a $1\times \omega$ row vector $\textbf{x} \in \mathbb{F}^{\omega}$. At a node $t$ in the network, the ensemble of received symbols from $In(t)$ is mapped to a symbol in $\mathbb{F}$ specific to each output link in $Out(t)$ and the symbol is sent on that link. The set of receivers is denoted by $\mathcal{T}= \{T_1, T_2, \ldots, T_{|\mathcal{T}|}\}$.

For a natural number $n$,  $[n]$ shall denote the set $\{1, 2 , \ldots, n\}$.

A sequence of links $e_1 ,e_2 ,\ldots,e_l$ is called a path leading to a node $t$ if $e_1 \in   In(s)$, $e_l \in In(t )$, and $(e_j ,e_{j+1})$ is an adjacent pair for all $1\leq j \leq l-1$. Two directed paths $P_1$ and $P_2$ in $G$ are \textit{edge-disjoint} if the two paths do not share a common link. For  $i\in[ |\mathcal{T}|]$, there exist $\omega$ edge-disjoint paths $P_{i1} ,P_{i2} ,\ldots,P_{i\omega}$ leading to the receiver $T_i$. If $\mid Out(s) \mid > \omega$, the $\omega$ imaginary edges of the source are assumed to come from a super source node, $s'$. The source node $s$ and its $\omega$ imaginary incoming edges are then included in $\mathcal{V}$  and $\mathcal{E}$ respectively.


\section{Background}
\label{background}
This section defines some terminologies in network coding and matroid theory from existing literature, which are used during exposition of results in the  sections to follow. A well acquainted  reader may skip this section and refer back for clarity if unfamiliar terms are encountered.

\subsection{Linear Network Code}
\begin{definition}[Global Description of a Linear Network Code (See \cite{ITRY}, Ch. $19$)]
	An $\omega$ dimensional linear network code on an acyclic network over a base field $\mathbb{F}$ consists of a scalar $k_{d,e}\in \mathbb{F}$ for every adjacent pair of links $(d,e)$ in the network as well as a column $\omega$-vector ${f_e}$ for every link $e$ such that:
	
	\begin{enumerate}
		
		\item{ $f_e=\sum_{d\in In(t)} k_{d,e} f_d $ for $e \in Out\left(t\right)$.}
		
		\item{The vectors \textbf{$f_e$} for the $\omega$ imaginary links $e\in In(s)$ form the standard basis of the vector space $\mathbb{F}^\omega$.	
}		
	\end{enumerate}
The vector \textbf{$f_e$} is called  \textit{global encoding kernel} for a link $e$.
\end{definition}

Initially, the source $s$ generates a message \textbf{x} as a row $\omega$-vector. The symbols in \textbf{x} are regarded as being received by source node $s$ on the $\omega$ imaginary links as \textbf{x.$f_d$}, $d \in In\left(s\right)$. Starting at the source node $s$, any node $t$ in the network receives the symbols \textbf{x.$f_d$}, $d\in In\left(t\right)$, from which it calculates the symbol \textbf{x.$f_e$} for sending on each link  $e\in Out \left(t\right)$ as:
$$ \textbf{x}.f_e=\textbf{x} \sum_{d \in In\left(t\right)} k_{d,e} f_d= \sum_{d \in In\left(t\right)} k_{d,e}\left(\textbf{x}.f_d\right).$$

\noindent In this way, the symbol \textbf{x}.$f_e$ is transmitted on any link $e$ in the network.

Given the local encoding kernels $(k_{d,e})$ for all the links in an acyclic network, the global encoding kernels can be calculated recursively in any upstream-to-downstream order. A linear network code can thus be specified by either the global encoding  kernels or  the local  encoding  kernels.


\subsection{Linear Multicast}
In a  multicast network, for a non-source node $t$ with $maxflow(t)\geq \omega$ ($i. e.$ every receiver node), there exist $\omega $ edge-disjoint paths from the $\omega$ imaginary incoming links of $s$ to $\omega$ distinct links in $In(t)$. A multicast network is solvable over $\mathbb{F}$ if and only if each receiver is able to decode $\omega$ messages produced by the source.

For a non-source node $t$, let the vector space generated by the global encoding kernels of the edges in $In(t)$ be denoted by $V_t$.
\begin{definition}[Linear Multicast\cite{ITRY}, Ch. $19$]
	\label{multicast}
	An  $\omega$-dimensional linear network code on an acyclic network qualifies as a linear multicast if $dim(V_t) = \omega$ for every non-source node $t$ with maxflow $\left(t\right)\geq \omega$.
	
\end{definition}

%
%



\subsection{Matroid}
\begin{definition}
	\label{matroid}
	A matroid $M$ (See \cite{Oxley}, Ch. $1$) is an ordered pair $(E, \mathcal{I})$ consisting of a finite set $E$ and a collection $\mathcal{I}$ of subsets of $E$ having the following three properties:
	\begin{enumerate}
		\item $\phi \in \mathcal{I}$.
		\item If $I \in \mathcal{I}$ and $I' \subseteq I$, then $I' \in \mathcal{I}$.
		\item If $I_1$ and $I_2$ are in $\mathcal{I}$ and $|I_1|< |I_2|$, then there is an element $e$ of $I_2- I_1$ such that $I_1 \cup e \in \mathcal{I}$.
	\end{enumerate}
\end{definition}

\subsection{Dual of a Matroid}
\begin{definition}
	\label{dual}
	Let $M$ be a matroid and $\mathcal{B}^*(M)$ be $\{E(M)-B: B \in \mathcal{B}(M)\}$, then $\mathcal{B}^*(M)$ is the set of bases of a matroid on $E(M)$, called the \textit{dual} of $M$. 
\end{definition}


\subsection{Graphic matroids}
\label{graphicmatroid}
\begin{definition}
	(Graphic matroid \cite{Oxley}, Ch. $5$) Let $E$ be the set of edges of a graph $G$ and $\mathcal{C}$ be the set of edge sets of cycles of $G$. Then $\mathcal{C}$ is the set of circuits of a matroid on $E$, called the \textit{cycle matroid} or \textit{polygon matroid} of $G$.
	
	Any matroid that is isomorphic to the cycle matroid of a graph is called a \textit{graphic matroid}.
\end{definition}

Graphic matroids form a fundamental class of matroids with numerous results and operations related to graphs having their counterparts for matroids. It is well known that every graphic matroid $M(G)$ is regular, i.e., $M(G)$ is representable over all finite fields.  Also, the dual of $M(G)$, $M^*(G)$ is also regular \cite{Oxley}.

Oxley \cite{Oxley} states that given a graph $G$, a representation for its graphic matroid $M(G)$ can be obtained by using the following construction. Form a directed graph $D(G)$ from $G$ by arbitrarily assigning a direction to each edge. Let $A_{D(G)}=[a_{ij}]$ denote the incidence matrix of $D(G)$, where

\begin{equation}
a_{ij} =
\begin{cases}
1, & \text{if vertex $i$ is the tail of non-loop arc $j$; },
\\
-1, & \text{if vertex $i$ is the head of non-loop arc $j$;},
\\
0, & \text{otherwise} .
\end{cases}
\end{equation}

Then $A_{D(G)}$  represents $M(G)$ over an arbitrary field $\mathbb{F}$ (refer \cite{Oxley} Lemma $5.1.3$ for proof).

\subsection{ Totally-unimodular Matrix}
\label{totumdisc}
\begin{definition}
	[Totally-unimodular matrix \cite{welsh}, Ch. $2$] A matrix of real numbers is totally unimodular if the determinant of every square submatrix is $1, -1$ or $0$.
\end{definition}

A matroid is regular if its representation matrix is totally unimodular (\cite{Oxley}, Ch. $6$). For a graphic matroid, as discussed before, $M(G)=M[A_{D(G)}]$. But $A_{D(G)}$ is a $(0, \pm1)$-matrix whose every column has at most one $1$ and one $-1$. In other words, $A_{D(G)}$ is totally unimodular (see \cite{Oxley} Ch. $5$, Lemma 5.1.4). This establishes the regularity of a graphic matroid and also its representation over an arbitrary field $\mathbb{F}$. These concepts shall be used in  the algorithm presented in Section \ref{algorithm} of this paper.

\subsection{Transversal Matroids}
\label{transversalmat}
Let $\mathcal{A}$ be a family $(A_1,A_2, \ldots, A_m)$ of subsets of a set $S$. A $\emph{transversal}$ or a system of distinct representatives of $(A_1,A_2, \ldots, A_m)$ is a subset $\{e_1,e_2,\ldots, e_m\}$ of $S$ such that $e_j \in A_j$ for all $j \in [m]$ and $e_1,e_2,\ldots, e_m$ are distinct. In other words, $T$ is a transversal of $(A_j:j\in J)$ if there is a bijection $\psi :J \rightarrow T$ such that $\psi(j) \in A_j$ for all $j$ in $J$.

If $X\subseteq S$, then $X$ is a \emph{partial transversal} of $A_j:j \in J$, if for some subset $K$ of $J$, $X$ is a transversal of $(A_j:j\in K )$. 

\begin{definition}[Transversal Matroid \cite{Oxley}, Ch. $1$] Let $\mathcal{I}$ be the set of partial transversals of $\mathcal{A}$. Then $\mathcal{I}$ is the collection of independent sets of a matroid on $S$ called the transversal matroid.
\end{definition}

Partial transversals can also be defined in terms of a matching in a bipartite graph.  The bipartite graph associated with $\mathcal{A}$ has vertex set $S \cup J$ and edge set $\{sj: s \in S, j\in J, s \in A_j\}$. A matching in a graph is a set of edges in the graph no two of which have a common endpoint.
A subset $X$ of $S$ is a partial transversal of $\mathcal{A}$ iff there is a matching in the bipartite graph in which every edge has one end-point in $X$ \cite{Oxley}.

\subsection{Gammoids}

\label{gammoids} 
The class of matroids knows as \textit{gammoids} was first discovered by Perfect \cite{perfect}. Though considered to be difficult to handle, gammoids are closely related to transversal matroids.

	Let $G=(V,E)$ be a directed graph with vertex set $V$ and edge set $E$. When $A, B \subset V$, there exists a \textit{linking} of $A$ onto $B$ if for some bijection $\alpha:A \rightarrow B$ we can find node-disjoint paths $(P_v: v\in A)$ in $G$ such that $P_v$ has initial vertex $v$ and terminal vertex $\alpha(v)$. 
\begin{definition}[Gammoids \cite{welsh}, Ch. $13$] 
	\label{gammoids_def}
Choose a fixed subset $B$ of $V$ and let $L(G,B)$ denote the collection of subsets of $V$ which can be linked to $B$. That is $X\in L(G,B)$ if there exists a $Y\subseteq B$ such that there is a linking of $X$ onto $Y$. $L(G,B)$ is the collection of independent sets of a matroid on $V$. 

Any matroid which can be obtained from a directed graph $G$ and some subset $B$ of $V(G)$ in this way, is called a \textit{strict gammoid}.
 A \textit{gammoid} is a matroid which is obtained by restricting a strict gammoid $L(G,B)$ to some subset of $V$.
\end{definition}
 A matroid is a strict gammoid if and only if its dual matroid is transversal (see \cite{welsh}, Ch. $13$, Theorem $2$). This duality shall be exploited in the proof of Lemma \ref{mainlemma}.

 \subsection{Induced Matroid} 
 \begin{definition}[Induced Matroid of a Linear Network Code\cite{OLMLNC}]
 	
 	Given a network $\mathcal{N}$ and a linear network code $L$ over $\mathbb{F}$. Let $M$ be the matrix obtained by juxtaposing all the $\mid\mathcal{E}\mid$ global encoding kernels of the network code. Let $\mathcal{I}$ be the family of sets of edges whose global encoding kernels are linearly independent.
 	Then $(\mathcal{E},\mathcal{I})$ is the matroid induced by the network code $L$ of $\mathcal{N}$ and is denoted by $M_{IN}(\mathcal{E},\mathcal{I}).$
 \end{definition}
 Among the class of induced matroids, the network matroid has the maximum family of independent sets\cite{OLMLNC}. The matroid induced by a generic linear network code \cite{ITRY} is a representation of the network matroid. In a linear multicast, not all bases of the network matroid need be the bases of the induced matroid.

 \subsection{Series-parallel Extension of a Matroid}
 \label{seriesparallel}
  Let $M$ be a matroid on $S$, let $x \in S$ and suppose $y \notin S$.
\begin{definition}[Series extension of a matroid (\cite{welsh}, Ch. $2$ )] The series extension of $M$ at $x$ by $y$ is the matroid $sM(x,y)$ on $S\cup y$ which has as its bases the sets of the form:\\
	$1)$ $B\cup y$; $B$ is a base of $M$, or\\
	$2)$ $B\cup x$; $B$ a base of $M$, $x \notin B$.
\end{definition}

\begin{definition}[Parallel extension of a matroid (\cite{welsh}, Ch.$2$ )] The parallel extension of $M$ is the matroid $pM(x,y)$ on $S\cup y$ which has as its bases the sets of the form:\\
	$1)$ $B$ is a base of $M$, or\\
	$2)$ $(B\backslash x)\cup y$; $x\in B$, $B$ a base of $M$.
	\label{parallel}
	
\end{definition}

\begin{definition}
	[Series-parallel extension of a matroid \cite{welsh}, Ch.$2$] A series-parallel extension of a matroid $M$ is a matroid which can be obtained from $M$ by successive series and parallel extensions. 
\end{definition}

\section{Multicast Matroid}
\label{multmatsection}
This section introduces and details the concept of multicast matroid in the context of a linearly solvable multicast network. 
\subsection{Defining a Multicast Matroid}

In a multicast network $\mathcal{N}$ of dimension $\omega$, each receiver $T_i\in \mathcal{T}$ is connected to the $\omega$ edges of the source through $\omega$ edge-disjoint paths 
For a network $\mathcal{N}$, consider any one receiver, say, $T_1$ and its $\omega$ edge-disjoint paths. Let the set of edges that form the $\omega$ edge-disjoint paths of $T_1$ be $\mathcal{E}_1$.

%
A bipartite graph $H(S,T,E)$ is formed  as follows:
\begin{IEEEeqnarray}{rCl}
	S & = &\big\{e_i ;\; e_i\in \mathcal{E}_1 \big\},
	\label{eq:block3_eq1}\\
	T & = &\big\{\hat{e}_i ;\; e_i \in \mathcal{E}_1 \backslash \{e_1,e_2, \ldots, e_\omega\}\big\},
	\label{eq:block3_eq2}\\
	E & = & \big\{(e_i,\hat{e}_i);e_i \in \mathcal{E}_1\backslash \{e_1,e_2, \ldots, e_\omega\}\big\}\bigcup\nonumber \\%
	  &  & \big\{(e_i, \hat{e}_j); e_i\in \mathcal{E}_1, \text{ $ e_j$ is the edge succeeding $e_i$}\nonumber\\
	  &   & \text{in one of the $\omega$ edge-disjoint  paths}\big\}.
	\label{eq:block3_eq3}
\end{IEEEeqnarray}

Here, the set $S$ is a copy of the edge set $\mathcal{E}_1$ and the set $T$ is a disjoint copy of $\mathcal{E}_1$ excluding the $\omega$  edges of $Out(s)$.
 
The bipartite graph obtained from the edge-disjoint paths of receiver $T_1$ for the network in Fig. \ref{fig1} is shown in Fig. \ref{fig3}.

\begin{figure}[htbp]
	\centering
	\includegraphics[scale=0.45]{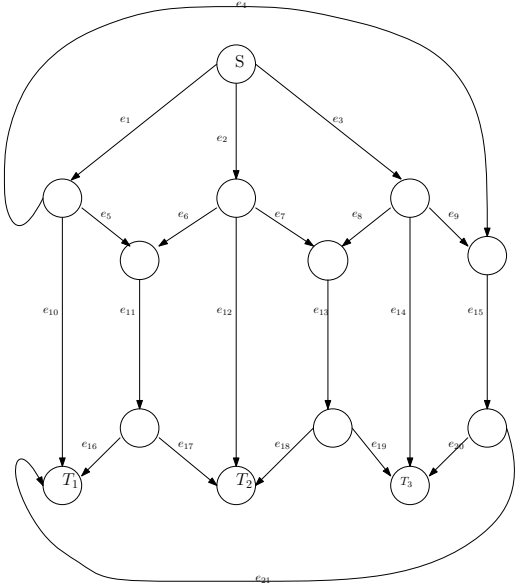}
	
	\caption{A single-source multicast network with three receivers}
	\label{fig1}
\end{figure}
\begin{figure}[htbp]
	\centering
	\includegraphics[scale=0.42]{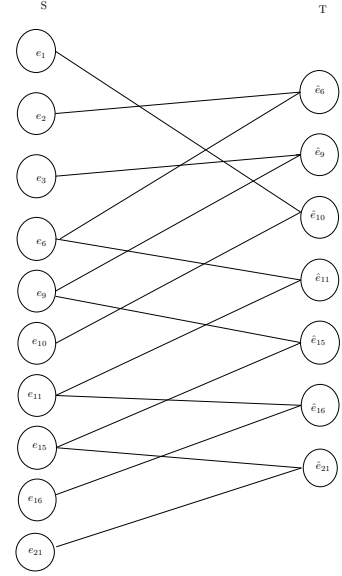}
	
	\caption{Bipartite graph $H$ constructed from edge-disjoint paths of $T_1$.}
	\label{fig3}
\end{figure}

In  the lemmas [\ref{if}, \ref{onlyif}], $\mathbb{T}$ denotes an arbitrary receiver $T_i\in \mathcal{T}$, $P_1, P_2, \ldots, P_{\omega}$ denote its $\omega$ edge-disjoint paths and $\mathbb{H}(S, T, E)$ is the corresponding bipartite graph.

\begin{lemma} 
	\label{if}
	If there is a perfect matching in $\mathbb{H}$ from the nodes in $S-B$ to $T$, there exist $\omega$ edge-disjoint paths from $s$ to the edges in $B$ of the network.
\end{lemma}
\begin{proof}
	Consider a set of nodes $S-B$ such that there exists a perfect matching in $\mathbb{H}$. Let $B=\{e_{a_1}, e_{a_2}, \ldots, e_{a_{|B|}}\}$. In $\mathbb{H}$, the nodes $\{e_1,e_2, \ldots, e_{\omega}\}\backslash B $ must necessarily be matched to the succeeding edge in the edge-disjoint paths $P_{1}, P_{2}, \ldots, P_{\omega}$ respectively. This is because, in $\mathbb{H}$ there is only one edge from each node in $\{e_1, e_2, \ldots, e_{\omega}\}$ to $T$. Let these nodes in $T$ be $\{\hat{e}_{i_1}, \hat{e}_{i_2}, \ldots, \hat{e}_{i_\omega}\}$.
	
	 Now, the nodes $\{{e}_{i_1}, {e}_{i_2}, \ldots, {e}_{i_\omega}\}\backslash B$ must be matched to the succeeding edge in paths $P_1, P_2, \ldots , P_{\omega}$  respectively. This is because, any node $e_{i_j}$ can be matched to either $\hat{e}_{i_j}$ or its succeeding edge in $P_{j}$ path. This would progressively continue until nodes $\hat{B}=\{\hat{e}_a, \hat{e}_b, \ldots, \hat{e}_{|B|}\}$ have been matched.  Thus, a perfect matching in $H$ from $S-B$ to $T$  always implies the existence of $\omega$ edge-disjoint paths from the source to edges in $B$.	
\end{proof}

\begin{lemma}
	\label{onlyif}
	If there are $\omega$ edge-disjoint paths in $\mathcal{N}$ from $s$ to the edges in $B$, there always exists a perfect matching in $\mathbb{H}$ from the nodes in $S-B$ to $T$.
\end{lemma}
\begin{proof}
	There are $\omega$ edge-disjoint paths from $s$ to $B$ in the network $\mathcal{N}$, given by $\mathbb{P}_1, \mathbb{P}_2, \ldots, \mathbb{P}_{\omega}$. A perfect matching $M$ from $S-B$ to $T$ can be constructed  as follows:
	\begin{IEEEeqnarray}{rCl}
		M & = & \big\{ \cup_{i=1}^{\omega} \{(e_j,\hat{e}_k) \in E: {e}_k \text{ is the edge succeeding $e_j$} \nonumber\\
		&    &\text{in path $\mathbb{P}_i$}\big\} \bigcup   \big\{(e_k,\hat{e}_k)\in E:e_k \notin \mathbb{P}_i, \forall i \in [\omega]\nonumber \big\} \\%
		\label{eq:block3_eq3}
	\end{IEEEeqnarray}

	Because any edge in $\mathcal{N}$ is in exactly one of the $\omega$ paths or none of them, $M$ is a perfect matching between $S-B$ and $T$ in $H(S,T,E)$.
\end{proof}

\begin{lemma}
	\label{dual}
	A matroid $M$ is a strict gammoid if and only if its dual matroid $M^*$ is transversal (see \cite{welsh}, Ch. $13$).	
\end{lemma}

\begin{lemma}
	\label{mainlemma}
	In a multicast network $(\mathcal{V}, \mathcal{E},s)$, consider any one receiver say $T_i$. Let the $\omega$ edge-disjoint paths from the imaginary edges of the source to $T_i$ be $P_{i1}, P_{i2}, \ldots, P_{i\omega}$. Let the subset of edges constituting these $\omega$ paths be $\mathcal{E}_i \subseteq \mathcal{E}$. Let $\mathcal{B}$ be the family of subsets $X \subseteq \mathcal{E}_i$, which satisfies that there are $\omega$ edge-disjoint paths from the source to edges in $X$. Then, $(\mathcal{E}_i,\mathcal{B})$ forms a  matroid $M_{G_i}$ which is a strict gammoid, with $\mathcal{B}$ as the set of bases.
\end{lemma}

\begin{proof}
	Using Lemmas $\ref{if}$ and $\ref{onlyif}$ it is clear that the matroid $M_{G_i}$ is the dual of the tranversal matroid induced from  the bipartite graph $H$, where $H(S,T,E)$ is constructed starting with the  edge-disjoint paths of receiver $T_i$. From Lemma \ref{dual}, $M_{G_i}$ is a strict gammoid.	
\end{proof}

\begin{theorem}
	\label{flinear}
	Given a multicast network $\mathcal{N}$ that is scalar linearly solvable over $\mathbb{F}$. Let $\mathcal{E}_i$ denote the set of edges in $\omega$ edge-disjoint paths of a receiver $T_i, i \in [|\mathcal{T}|]$. The vector matroid $M_{IN}$ induced by an $\mathbb{F}$-linear network code, when restricted to the edges in $\mathcal{E}_i$   gives an $\mathbb{F}$-representation of the gammoid $\mathcal{M}_{G_i}$.
\end{theorem}
\begin{proof}
	By definition, in a multicast network of dimension $\omega$ there are $\omega$ edge-disjoint paths from the source to each of the receivers. Consider these $\omega$ paths from  the source $s$ to any of the receivers, say, $T_i$. Let $X\subset \mathcal{E}_i$ be the set of $\omega$ edges that  constitute a cut along these edge-disjoint paths separating $s$ and $T_i$.  Any  $\mathbb{F}$-linear multicast  exists only if, the $\omega$  global encoding kernels $f_{e_i}; \forall\; e_i\in X$ are linearly independent. 
	This  precisely satisfies the requirements for bases of the gammoid $M_{G_i}$.
\end{proof}

%
%
%
Next, a new matroid called the multicast matroid is defined. As discussed in Lemma \ref{mainlemma}, we have $|\mathcal{T}|$ strict gammoids $M_{G_1}, M_{G_2}, \ldots, M_{G_{|T|}}$ in the network $\mathcal{N}$.
Also recall the definition of parallel extension of a matroid (Definition \ref{parallel}).
\begin{definition}[Multicast Matroid]
	\label{multmatdef} 
	Given a linear multicast over a base field $\mathbb{F}$, on a multicast network $\mathcal{N}$.
	Let the $\omega$ edge-disjoint paths from $s$ to receiver $T_2$ be $P_{21},P_{22}, \ldots, P_{2\omega}$. Let $B_1\subset \mathcal{E}_2$ be such that $f_{B_1}=\{f_{e_{k_1}}, f_{e_2}, \ldots, f_{e_{\omega}}\}$, where $e_{k_1}$ is the edge succeeding $e_1$ in path $P_{21}$. We know that, $f_B=\{f_{e_1}, f_{e_2}, \ldots, f_{e_{\omega}}\}$ is a basis of the matroid $M_{G_1}$.
	$$f_{B_1}=(\{f_{e_1}, f_{e_2}, \ldots, f_{e_{\omega}}\} \backslash f_{e_1}) \cup f_{e_{k_1}}.$$ 
	Clearly, $f_{B_1}$ was obtained by parallel extension of the basis $f_B$ at $e_1$. By extending $f_B$ at each of $e_2, e_3, \ldots, e_{\omega}$ we get $f_{B_i}=(f_B\backslash f_{e_i} )\cup f_{e_{k_i}}, i\in 2,3 \ldots, \omega$; where ${e_{k_i}}$ is the edge succeeding $e_i$ in path $P_{2i}$.   
	In this way, every other basis in $M_{G_2}$ can be obtained progressively by extending the bases  $f_{B_1}, f_{B_2}, \ldots, f_{B_{\omega}}$.
	
	Likewise, the bases of each of the gammoids $M_{G_3}, M_{G_4}, \ldots M_{G_{|\mathcal{T}|}}$  can be obtained by parallel extension. 
	
	The matroid $\mathcal{M}$ which is so obtained as a  parallel extension of the strict gammoid $M_{G_1}$ (given  in Lemma \ref{mainlemma}) is called the \textit{multicast matroid}. The ground-set of $\mathcal{M}$ is the set of all global encoding kernels of $\mathcal{N}$, $f_{\mathcal{E}}=\{f_{e_1}, f_{e_2}, \ldots, f_{e_{|\mathcal{E}|}}\}$. The rank of $\mathcal{M}$ is $\omega$.
\end{definition}
\begin{theorem}
	\label{A1}
	Multicast matroid is a gammoid.
\end{theorem}

\begin{proof}
	A series-parallel extension of a gammoid is a gammoid (see \cite{welsh}, Ch. $14$). By Definition \ref{multmatdef}, the multicast matroid is a parallel extension of the strict gammoid $M_{G_1}$. Hence,  multicast matroid is a gammoid.
\end{proof}

\begin{example}
	\label{ex1}
		For the butterfly network in Fig. \ref{figex1}, we have two strict gammoids $M_{G_1}$ and $M_{G_2}$. 
			
		\begin{figure}[htbp]
			\centering
			\includegraphics[scale=0.65]{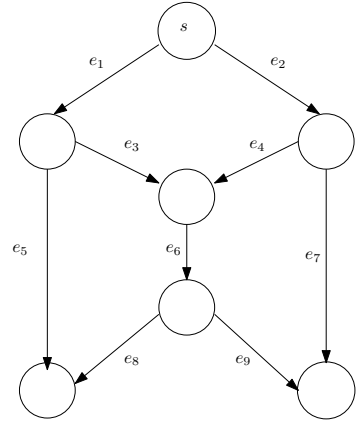}	
			\caption{A butterfly network}
			\label{figex1}
		\end{figure}
	For $\mathcal{M}_{G_1}(\mathcal{E}_1, \mathcal{B}_1)$:
\begin{IEEEeqnarray}{rCl}
	\mathcal{E}_1&=&\{e_1, e_2, e_4, e_5, e_6, e_8\}\nonumber\\
	\mathcal{B}_1 &= &\big\{\{{e_1}, {e_2}\}, \{{e_1}, e_4\}, \{e_1, e_6\}, \{e_1, e_8\}, \nonumber  \\
	 & & \{e_5, e_2\}, \{e_5, e_4\}, \{e_5, e_6\}, \{e_5, e_8\}\big\}.
\end{IEEEeqnarray}
%

	For $\mathcal{M}_{G_2}(\mathcal{E}_2, \mathcal{B}_2)$:
\begin{IEEEeqnarray}{rCl}
	\mathcal{E}_2&=&\{e_1, e_2, e_3, e_6, e_7, e_9\}\nonumber\\
	\mathcal{B}_2 &= &\big\{\{e_1, e_2\}, \{e_3, e_2\}, \{e_6, e_2\}, \{e_9, e_2\}, \nonumber  \\
	& &  \{e_1, e_7\}, \{e_3, e_7\},\{e_6, e_7\}, \{e_9, e_7\}\big\}.
\end{IEEEeqnarray}

For the multicast matroid $\mathcal{M}(f_{\mathcal{E}}, \mathcal{B})$,
\begin{IEEEeqnarray}{rCl}
		\mathcal{B}&=&\big\{\{f_{e_1}, f_{e_2}\}, \{f_{e_1}, f_{e_4}\}, \{f_{e_1}, f_{e_6}\}, \{f_{e_1}, f_{e_8}\}, \{f_{e_5}, f_{e_2}\},\nonumber  \\
		& &  \{f_{e_5}, f_{e_4}\}, \{f_{e_5}, f_{e_6}\}, \{f_{e_5}, f_{e_8}\},   \{f_{e_3}, f_{e_2}\}, \{f_{e_6}, f_{e_2}\}, \nonumber\\
		 &&  \{f_{e_9}, f_{e_2}\},   \{f_{e_1}, f_{e_7}\}, \{f_{e_3}, f_{e_7}\}, \{f_{e_6, e_7}\}, \{f_{e_9}, f_{e_7}\}\big\}.\nonumber  
\end{IEEEeqnarray}
\end{example}
\section{Base Orderability of the Multicast Matroid}
\label{baseordmult}
\begin{definition}[Base orderable matroid\cite{welsh}]
	A matroid $M$ on $S$ is base-orderable if for any two bases $B_1$, $B_2$ of $M$ there exists a bijection $\pi:B_1 \rightarrow B_2$
	such that both $(B_1\setminus x ) \cup \pi(x)$ and $(B_2\setminus \pi(x))\cup x$ are  bases of $M$ for each $x \in B_1$. Such a $\pi$ is called an exchange ordering for $B_1, B_2$.
\end{definition}

%
%
%


\begin{lemma}
	\label{gammbase}
	Gammoids are base orderable. (See \cite{welsh}, Ch.$14$, Th.$1$).
\end{lemma}
\begin{theorem}
	\label{dubase}
	The multicast matroid $\mathcal{M}$ is base-orderable.
\end{theorem}
\begin{proof}
	Using Theorem \ref{A1}, the multicast matroid $\mathcal{M}$ is a gammoid. Using lemma \ref{gammbase}, gammoids are base-orderable. Hence proved.
%
\end{proof}	

\section{Multicast Network Solvable Over GF(2) and its Multicast  Matroid}
\label{fundaresults2}

In this section, the main result of the paper is given. Starting with a multicast network $\mathcal{N}(\mathcal{V}, \mathcal{E})$ that is solvable over $GF(2)$ it is proved that $\mathcal{N}$  is solvable over all fields. 



\begin{lemma}
	\label{binbasegraphic}
	If a matroid is binary and base-orderable, it is graphic
	\cite{welsh} .
\end{lemma}
Thus, given a multicast network that is solvable over $GF(2)$, we have the following theorem regarding its multicast matroid.
\begin{theorem}
	\label{graphic}
	The multicast  matroid $\mathcal{M}$  of a network that is scalar linearly solvable over $GF(2)$  is graphic. 
\end{theorem}
\begin{proof}
	Given a multicast network $\mathcal{N}$ that is solvable over $GF(2)$. Thus, global encoding kernels in the linear multicast over $GF(2)$ gives a representation for the multicast matroid $\mathcal{M}$. This implies that $\mathcal{M}$ is is now a binary matroid. Using theorem \ref{dubase}, $\mathcal{M}$ is base-orderable. Hence, using lemma \ref{binbasegraphic}, $\mathcal{M}$ is graphic when $\mathcal{N}$ is scalar linearly solvable over $GF(2)$.
\end{proof}	

\begin{lemma}
	\label{regular}
	The graphic matroid  of a graph $G$ is representable over all fields and  hence is regular  \cite{Oxley}.
\end{lemma}

\begin{theorem}
	A multicast network sovable over $GF(2)$ is solvable over all higher fields.
\end{theorem}

\begin{proof}
	Given a multicast network $\mathcal{N}$ that is solvable over $GF(2)$.  Using theorem \ref{graphic}, its multicast matroid is graphic and hence regular.
	
	Thus, $\mathcal{N}$ is matroidal with respect to a regular matroid $\mathcal{M}$ and hence scalar linearly solvable over all fields (\cite{Kim}). 
	The columns in an  $\mathbb{F}$-representation of $\mathcal{M}$ serve as global encoding kernels $f_{e_i}$s of the corresponding  $\mathbb{F}$-linear solution of $\mathcal{N}$.
\end{proof}
\section{Obtaining the solution over an arbitrary finite field }
\label{algorithm}

%
%
%

This section details an algorithm to find the linear network coding solution over an arbitrary finite field starting with a linear multicast over $GF(2)$.

The discussion on graphic matroids in Section \ref{graphicmatroid} shows how a binary representation matrix with at most two non-zeros per column is obtained for any graphic matroid.
\begin{definition}
	A binary matrix B is \textit{graphic} if it can be transformed using elementary row operations to a matrix that has at most 2 non-zeros per column\cite{Oxley}. 
\end{definition}

Section \ref{totumdisc} discusses the concept of totally-unimodular matrices and their close relation with graphic matroids.

\begin{lemma}
	\label{totum}
	A $(0,\pm 1)-matrix$ is a real matrix with every entry in $\{0, 1, -1\}$. Let $A$ be a $(0, \pm1)$-matrix whose every column has at most one $1$ and one $-1$. Then $A$ is totally unimodular (See \cite{Oxley}, Ch. $5$ ).
\end{lemma}

Given a multicast network $\mathcal{N}(\mathcal{V}, \mathcal{E},s)$ with its solution over $GF(2)$, the linear solution over an arbitrary field $\mathbb{F}$ is obtained as follows:
\begin{enumerate}
	\item Write the $\arrowvert\mathcal{E}\arrowvert$ global encoding kernels $f_e; e \in \mathcal{E}$ in juxtaposition to obtain the binary matrix $B$.
	
	\item  	From Theorem \ref{graphic}, the multicast matroid $\mathcal{M}$ of a network solvable over $GF(2)$ is graphic. Hence, any representation matrix  $B$, for the matroid $\mathcal{M}$ is graphic. Hence, perform row operations to get a  representation matrix $B'$ with at most two $1$s in every column.
	
	\item Sign the binary matrix $B'$ to get a totally unimodular matrix $\mathbf{B}$. This can be done by simply ensuring that each column of $\mathbf{B}$ has at most one $1$ and one $-1$.
	\item  	View the matrix $\mathbf{B}$ over an arbitrary field $\mathbb{F}$. Each column of $\mathbf{B'}$ gives the global encoding kernel of  its corresponding edge, in an $\mathbb{F}$-linear solution of the network $\mathcal{N}$. 
	
\end{enumerate}

\begin{example}
	The solution over $GF(2)$ for the network in Fig.\ref{fig1} is shown in Fig.\ref{fig2}.
	Juxtaposing the global encoding kernels, we get the following representation matrix for $\mathcal{M}$:
	\[B=\scalemath{0.7}{\left[\begin{array}{ccccccccccccccccccccc}
		1 & 0 & 0 & 1 & 1 & 0 & 0  & 0 & 0 & 1 & 1 & 0 & 0& 0 & 1 & 1 & 1 & 0 & 0  & 1 & 1\\
		0 & 1 & 0 & 0 & 0 & 1 & 1  & 0 & 0 & 0 & 1 & 1 & 1& 0 & 0 & 1 & 1 & 1 & 1  & 0 & 0\\
		0 & 0 & 1 & 0 & 0 & 0 & 0  & 1 & 1 & 0 & 0 & 0 & 1& 1 & 1 & 0 & 0 & 1 & 1  & 1 & 1\\
		\end{array}\right]}.\]

	\begin{figure}[htbp]
		\centering
		\includegraphics[scale=0.45]{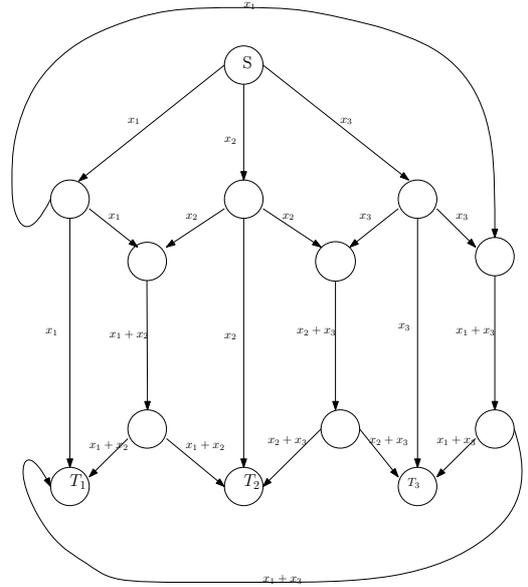}
		
		\caption{A single-source multicast network with three receivers}
		\label{fig2}
	\end{figure}

	A signing of the matrix $B$ to get a real totally unimodular matrix $\mathbf{B}$ is shown below:
	
	$$\scalemath{0.63}{\left[  \begin{array}{ccccccccccccccccccccc}
			1 & 0 & 0 & 1 & 1 & 0 & 0  & 0 & 0 & 1 &  1 & 0 &  0 & 0 &  1 &  1 &  1 & 0  & 0   &  1  & 1\\
			0 & 1 & 0 & 0 & 0 & 1 & 1  & 0 & 0 & 0 & {-1} & 1 &  1 & 0 &  0 & -1 & -1 & 1  &  1  &  0  & 0\\
			0 & 0 & 1 & 0 & 0 & 0 & 0  & 1 & 1 & 0 &  0 & 0 & -1 & 1 & -1 &  0 &  0 & -1 & -1  &  -1 & -1\\
			\end{array}\right].}$$

	Suppose, we need the solution over $\mathbb{F}_5$. When viewed over $\mathbb{F}_5$, the matrix $\mathbf{B}$ becomes:
	$$\mathbf{B}'=\scalemath{0.7} {\left[ \begin{array}{ccccccccccccccccccccc}
		1 & 0 & 0 & 1 & 1 & 0 & 0  & 0 & 0 & 1 &  1 & 0 &  0 & 0 &  1 &  1 &  1 & 0  & 0   &  1  & 1\\
		0 & 1 & 0 & 0 & 0 & 1 & 1  & 0 & 0 & 0 & {4} & 1 &  1 & 0 &  0 & 4 & 4 & 1  &  1  &  0  & 0\\
		0 & 0 & 1 & 0 & 0 & 0 & 0  & 1 & 1 & 0 &  0 & 0 & 4 & 1 & 4 &  0 &  0 & 4 & 4  &  4 & 4\\
		\end{array}\right].}$$
	Each column of $\mathbf{B}'$ represents the global encoding kernel for every edge in $\mathcal{N}$ for a solution over $\mathbb{F}_5$.
\end{example}

\section{Conclusion}
\label{conclusion}

This paper proves the conjecture that  if a multicast network is linearly solvable over GF(2) it is solvable over all higher finite fields. Also, given a multicast network that is solvable over GF(2), an algorithm to obtain the solution over an arbitrary field $\mathbb{F}$ from the solution over $GF(2)$ has been  worked out.


%

\balance
\end{document}